\documentclass{article}

\usepackage{fullpage}
\usepackage{authblk}

\usepackage[american]{babel}	
\usepackage[utf8]{inputenc}		
\usepackage[T1]{fontenc}	
\usepackage[english=american]{csquotes}	
\usepackage{verbatim}

\usepackage{url}
\usepackage[hidelinks]{hyperref}
\usepackage{amsmath,amssymb,amsthm}
\usepackage{mathtools}
\usepackage{xcolor}
\usepackage{graphics}
\usepackage[font={small},format=plain,labelfont=bf]{caption}
\usepackage[numbers]{natbib}
\usepackage{enumerate}

\usepackage[normalem]{ulem}

\theoremstyle{definition} 
\newtheorem{Def}{Definition}[section]
\newtheorem{Theo}[Def]{Theorem}
\newtheorem{Prop}[Def]{Proposition}
\newtheorem{Lem}[Def]{Lemma}
\newtheorem{Cor}[Def]{Corollary}
\newtheorem{Rem}[Def]{Remark}

\newcommand\blfootnote[1]{%
  \begingroup
  \renewcommand\thefootnote{}\footnote{#1}%
  \addtocounter{footnote}{-1}%
  \endgroup
}

\newcommand{\LS}{\mathcal{LS}}

\title{How far is my network from being edge-based? Proximity measures for edge-basedness of unrooted phylogenetic networks}

\author[1,$\ast$]{Mareike Fischer}
\author[1]{Tom Niklas Hamann}
\author[2]{Kristina Wicke}

\affil[1]{Institute of Mathematics and Computer Science, University of Greifswald, Greifswald, Germany}
\affil[2]{Department of Mathematics, The Ohio State University, Columbus OH, USA}

\date{}

\begin{document}

\maketitle

\begin{abstract}
Phylogenetic networks which are, as opposed to trees, suitable to describe  processes like hybridization and horizontal gene transfer, play a substantial role in evolutionary research. However, while non-treelike events need to be taken into account, they are relatively rare, which implies that biologically relevant networks are often assumed to be similar to trees in the sense that they can be obtained by taking a tree and adding some additional edges. This observation led to the concept of so-called tree-based networks, which recently gained substantial interest in the literature. Unfortunately, though, identifying such networks in the unrooted case is an NP-complete problem. Therefore, classes of networks for which tree-basedness can be guaranteed are of the utmost interest.

The most prominent such class is formed by so-called edge-based networks, which have a close relationship to generalized series-parallel graphs known from graph theory. They can be identified in linear time and are in some regards biologically more plausible than general tree-based networks. While concerning the latter proximity measures for general networks have already been introduced, such measures are not yet available for edge-basedness. This means that for an arbitrary unrooted network, the \enquote{distance} to the nearest edge-based network could so far not be determined. The present manuscript fills this gap by introducing two classes of proximity measures for edge-basedness, one based on the given network itself and one based on its so-called leaf shrink graph ($\LS$ graph). Both classes contain four different proximity measures, whose similarities and differences we study subsequently.
\end{abstract}

{\it Keywords:} phylogenetic network; tree-based network; edge-based network; GSP graph; $K_4$-minor free graph \\

\blfootnote{$^\ast$Corresponding author\\ \textit{Email address:} \url{email@mareikefischer.de}}

\section{Introduction}
Traditionally, phylogenetic trees were used to describe the relationships between different species. But unfortunately, trees cannot be used to describe  evolutionary events like horizontal gene transfer or hybridization  \citep{Amaral2014,Cui2013,Fontaine2015}, as these introduce cycles in the underlying graph. Therefore, phylogenetic networks were introduced and are nowadays widely acknowledged to be better descriptors of evolution than trees.

On the other hand, though, non-treelike events (so-called reticulation events) are for most species relatively rare, which implies that networks with fewer reticulations are usually biologically more plausible than networks which contain plenty of them. Therefore, researchers have been suggesting different concepts for \enquote{tree-like} networks, i.e., networks, which are still very similar to trees. Such concepts include, but are not limited to, level-$1$ networks (also known as galled trees) \citep{Choy2005,Gusfield2003,Jansson2006}, tree-child networks \citep{Cardona2009}, and tree-based networks \citep{Francis2015, Francis2018}. Particularly the latter have gained significant interest in recent literature
\citep{Fischer2020, FISCHER2021a, Francis2018a, Hendriksen2018, Jetten2018}.

Tree-based networks can be thought of as trees to which a few edges have been added, so they contain a so-called  \enquote{support tree}. The degree-1 vertices (\enquote{leaves}) of such a network coincide with the leaves of said support tree, meaning that the support tree is simply a spanning tree with the same number of leaves as the network. A network is called tree-based if it contains such a spanning tree. Such networks are interesting for the above mentioned biological reasons, but they are also of high interest from a graph theoretical point of view, for instance because finding spanning trees with as few leaves as possible is a graph theoretic problem appearing also in other contexts \citep{Bousquet2020,Salamon2008} and because their identification can be shown to be NP-complete by a reduction from the Hamiltonian path problem \citep{Francis2018}.

However, due to the NP-completeness of tree-basedness, this concept's value for practical applications seems limited. This subsequently led to the study of classes of networks for which tree-basedness can be guaranteed \citep{Fischer2020}. The most prominent such class are edge-based networks. When they were first introduced, their appeal seemed twofold: first, they are guaranteed to be tree-based, and second, they are so far the only known such class of networks which can be identified in linear time (thanks to their close relationship to  so-called generalized series-parallel (GSP) graphs, a well-known concept from graph theory \citep{Ho1999}).

Notwithstanding their merits, edge-based networks are often viewed as a mere subclass of the biologically relevant class of tree-based networks. However, in the present manuscript we argue that this view might be flawed. In fact, edge-basedness might make networks even more biologically plausible than tree-basedness. This is due to the fact that adding more and more edges makes a network more likely to be tree-based (as it increases the chance of containing a spanning tree with few leaves \citep{Fischer2020a}), whereas -- as we will show -- deleting edges can make a non-edge-based network edge-based (and adding edges can never achieve that). Regarding the above mentioned original intention of introducing tree-based networks, which was to develop a concept that would capture networks that are essentially similar to trees (and should thus, intuitively, contain only few reticulations), edge-based networks might thus be the more suitable concept.

Hence, edge-based networks are both of mathematical and biological interest. But many networks are not edge-based, and it is therefore interesting to determine their \enquote{distance} from the nearest edge-based network. In the present manuscript, we introduce two classes of proximity measures which can be used for this purpose: while the first class acts on the network itself, the second class acts on its so-called $\LS$ graph. We subsequently discuss the similarities and differences between the measures introduced in this study and analyze the computational complexity of computing them. We conclude our manuscript with a brief discussion of our results and by highlighting various possible directions of  future research.

\section{Preliminaries}

In this section, we introduce all concepts relevant for the present manuscript. We start with some general definitions.

\subsection{Definitions and notation}
\paragraph{Basic graph-theoretical concepts.}
Throughout this paper, $G=(V(G),E(G))$ (or $G=(V,E)$ for brevity) will denote a finite graph with vertex set $V(G)$ and edge set $E(G)$.
Note that graphs in this manuscript may contain parallel edges and loops. Whenever we require graphs without parallel edges and/or loops, we will refer to them as \emph{simple graphs}, and whenever parallel edges are allowed but loops are not, we will use the term \emph{loopless} graphs. In this manuscript, the \emph{order} of a graph $G$ is defined as $\lvert V(G)\rvert$. 

If $G=(V,E)$ is a graph and $V' \subseteq V$ is a subset of its vertices, then the \emph{induced subgraph} $G[V']$ is the graph whose vertex set is $V'$ and whose edge set consists of all edges of $E$ with both endpoints in $V'$.

In the following, it will be useful to consider decompositions of connected graphs. Therefore, let $G=(V,E)$ be a connected graph. A \emph{cut edge}, or \emph{bridge}, of $G$ is an edge $e$ whose removal disconnects the graph. Similarly, a vertex $v$ is a \emph{cut vertex}, or \emph{articulation}, if deleting $v$ and its incident edges disconnects the graph. A \emph{blob} in a connected graph $G$ is a maximal connected subgraph of $G$ that has no cut edge. Note, however, that a blob may contain cut vertices. If a blob consists only of one vertex, we call the blob \emph{trivial}.
A \emph{block}, on the other hand, is a maximal biconnected subgraph of $G$, i.e., a maximal induced subgraph of $G$ that remains connected if any of its vertices is removed. In particular, a block does not contain cut vertices. Note that for technical reasons the complete graph $K_2$, i.e., a single edge, is considered to be a biconnected graph as well.

Another graph-theoretical concept relevant for the present manuscript is the notion of minors. A graph $G'=(V',E')$ is a \emph{minor} of a graph $G=(V,E)$ if $G'$ can be obtained from $G$ by a series of vertex deletions, edge deletions, and/or edge contractions. Moreover, $G'$ is called a \emph{topological minor} of $G$ if a subdivision of $G'$ is isomorphic to a subgraph of $G$. Here, a \emph{subdivision} of a graph $G$ is a graph resulting from subdividing edges of $G$, where \emph{subdividing} an edge $e$, $\{u,v\}$ say, refers to the process of deleting $e$, adding a new vertex $w$, and adding the edges $\{u,w\}$ and $\{w,v\}$. Note that every topological minor is also a minor \cite[Proposition 1.7.3]{Diestel2017}, whereas the converse is not true in general. If a graph $G$ does not contain $G'$ as a (topological) minor, we say that $G$ is \emph{$G'$-(topological) minor free}.

\paragraph{Phylogenetic networks and related concepts.}
Let $X$ be a non-empty finite set (e.g., of taxa or species). An \emph{unrooted phylogenetic network on $X$} is a connected simple graph $N=(V,E)$ without degree-2 vertices, where the set of degree-1 vertices (referred to as the \emph{leaves}) is bijectively labeled by and identified with $X$. Such an unrooted phylogenetic network is called an \emph{unrooted phylogenetic tree} if the underlying graph structure is a tree.
Note that we do not impose any additional constraints on the non-leaf vertices of $N$ in this manuscript; in particular, we do not restrict the analysis to unrooted binary phylogenetic networks in which each interior vertex $v \in V \setminus X$ has degree precisely 3.  For the remainder of the paper, we will refer to unrooted phylogenetic networks and unrooted phylogenetic trees as phylogenetic networks and phylogenetic trees, respectively, as we only consider unrooted ones.

Following \cite{FISCHER2021a}, we call a phylogenetic network $N$ \emph{proper} if the removal of any cut edge or cut vertex present in the network leads to connected components containing at least one element of $X$ each.

Moreover, following \cite{Fischer2020a}, we say that a phylogenetic network $N$ has \emph{tier $k$}, if $k$ is the minimum number of edges of $N$ whose deletion turns $N$ into a tree. Note that the tier does not depend on $N$ being a phylogenetic network and can be defined analogously for connected graphs. A related concept is the \emph{level} of a phylogenetic network. More precisely, $N$ is said to have \emph{level $k$} or to be a \emph{level-$k$} network if at most $k$ edges have to be removed from each blob of $N$ to obtain a tree. In other words, $N$ is a level-$k$ network, if the maximal tier of the blobs of $N$ is $k$. In particular, for any network $N$, level$(N)\leq$ tier$(N)$.

Finally, a phylogenetic network $N=(V,E)$ on $X$ is called \emph{tree-based} if there is a spanning tree $T=(V,E')$ in $N$ (with $E' \subseteq E$) whose leaf set is equal to $X$. 
Note that not every phylogenetic network is tree-based and deciding whether an unrooted phylogenetic network is tree-based is an NP-complete problem \cite{Francis2018}. However, a necessary condition for a network $N$ to be tree-based is that $N$ is proper \cite{FISCHER2021a}. Examples of three proper phylogenetic networks, two of them tree-based, one of them non-tree-based, are shown in Figure \ref{Fig_3Networks}.

\paragraph{Edge-based phylogenetic networks and (generalized) series-parallel graphs.}
We now recall the so-called \emph{leaf-shrinking procedure} from \cite{Fischer2020,Fischer2020Correction}. Let $G=(V,E)$ be a connected graph with at least two vertices, i.e., $|V(G)| \geq 2$. Then, the \emph{leaf shrink graph} ($\LS$ graph for short) $\LS(G)=(V_\LS,E_\LS)$ is the unique graph obtained from $G$ by performing an arbitrary sequence of the following operations until no further reduction is possible and such that each intermediate (and the final) graph has at least two vertices:
\begin{enumerate}[(i)]
    \item Delete a leaf (i.e., a degree-1 vertex) and its incident edge;
    \item Suppress a degree-2 vertex;
    \item Delete one copy of a multiple (also called parallel) edge, i.e., if $e_1 = e_2 \in E(G)$, delete $e_2$;
    \item Delete a loop, i.e., if $e=\{u,u\} \in E(G)$, delete $e$.
\end{enumerate}
Note that the smallest graph (in terms of the number of vertices and the number of edges) a graph $G$ may be reduced to is the complete graph on two vertices $K_2$, i.e., a single edge. This motivates the following definition adapted from \cite{Fischer2020,Fischer2020Correction}.

\begin{Def}
    Let $G$ be a connected graph with $|V(G)|\geq 2$. If the leaf shrink graph $\LS(G)$ of $G$ is a single edge, $G$ is called \emph{edge-based}.
    If $G=N$ is a phylogenetic network on $X$ with $|X| \geq 2$, and $\LS(N)$ is a single edge, $N$ is called an \emph{edge-based} network.
\end{Def}

\noindent Note that the original definition of edge-based networks given in \cite{Fischer2020, Fischer2020Correction} required the network to be proper; however, we will later on show that every edge-based network is proper (cf. Corollary \ref{cor_eb=proper}).

\begin{figure}
    \centering
    \includegraphics[scale=0.4]{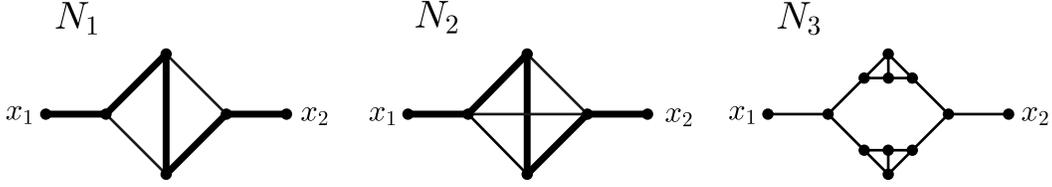}
    \caption{Three proper phylogenetic networks $N_1, N_2$, and $N_3$ on $X = \{x_1,x_2\}$. Network $N_1$ is edge-based and tree-based, whereas $N_2$ is tree-based but not edge-based (in both cases, a spanning tree $T$ with leaf set equal to $X$ is highlighted in bold). Network $N_3$ is not tree-based and thus in particular not edge-based.}
    \label{Fig_3Networks}
\end{figure}

Moreover, note that edge-based networks are closely related to a well-known family of graphs, namely the family of generalized series-parallel graphs.
Therefore, recall that a connected and loopless graph $G$ is called a \emph{generalized series-parallel} (GSP) graph if it can be reduced to $K_2$ by a series of operations of types (i), (ii), and (iii), i.e., by deleting degree-1 vertices, suppressing degree-2 vertices, and deleting copies of parallel edges. If $G$ can be reduced to $K_2$ by only using operations of types (ii) and (iii), it is called a \emph{series-parallel} (SP) graph. Note that there is the following close connection between GSP and SP graphs:

\begin{Lem}[{adapted from \cite[Lemma 3.2]{Ho1999}}] \label{lem:GSP_blocks_SP}
A connected graph $G$ is a GSP graph if and only if each block of $G$ is an SP graph.
\end{Lem}

Comparing the definitions of GSP graphs and edge-based graphs, there is seemingly a slight difference between the two classes. Specifically, both can be reduced to $K_2$ by certain restriction operations; however, the deletion of loops is a valid restriction operation in the case of edge-based graphs, but not in the case of GSP graphs. Nevertheless, there is a direct relationship between both classes of graphs.

\begin{Lem} [{\cite[Corollary 1]{Fischer2020}}] \label{lem:GSP_edgebased}
Let $G$ be a connected graph. Then $G$ is a GSP graph if and only if it is loopless and edge-based.
\end{Lem}

As every phylogenetic network is loopless by definition, this in particular implies that every edge-based phylogenetic network is a GSP graph. As GSP graphs can be recognized in linear time \cite{Ho1999}, this in turn implies that edge-based phylogenetic networks can be recognized in linear time. Finally, as shown in \cite[Theorem 3]{Fischer2020}, every edge-based phylogenetic network is tree-based, and thus edge-based phylogenetic networks constitute a class of tree-based networks that can be recognized in linear time.

\subsection{Known results}
\noindent Next we state some results known in or easily derived from the literature which we need for the present manuscript. On the one hand, these concern properties and characterizations of GSP and edge-based graphs and networks. Note that for technical reasons, we state some of these results concerning edge-basedness for general graphs rather than phylogenetic networks.
On the other hand, we recall some results concerning the computational complexity of determining the minimum number of graph operations (e.g., vertex deletions, edge deletions, or edge contractions) that need to be performed to turn a given graph $G$ into a graph $G'$ satisfying a certain property. The latter will allow us to assess the computational complexity of computing the proximity measures for edge-basedness of phylogenetic networks introduced subsequently.

\subsubsection{Properties of GSP graphs and edge-based graphs and networks}
We begin by giving an alternative characterization of edge-based graphs.

\begin{Prop}\label{prop:K4-minorfree}
Let $G=(V,E)$ be a connected graph with $|V| \geq 2$. Then, $G$ is edge-based if and only if $G$ is $K_4$-minor free.
\end{Prop}

In order to prove this proposition, we require the following two statements from \cite{Bernshteyn2021} and \cite{Diestel2017}, respectively.

\begin{Prop}[{adapted from \cite[Corollary 8.5]{Bernshteyn2021}}] \label{prop:K4_GSP}
Let $G$ be a connected graph that does not contain $K_4$ as a topological minor. Then, $G$ is a GSP graph.
\end{Prop}

\begin{Lem}[{direct consequence of \cite[Proposition 1.7.3]{Diestel2017}}]\label{lem:k4minor_top}
Let $G$ be a graph. Then, $G$ contains $K_4$ as a minor if and only if $G$ contains $K_4$ as a topological minor.
\end{Lem}

We are now in the position to prove Proposition \ref{prop:K4-minorfree}.
\begin{proof}[Proof of Proposition \ref{prop:K4-minorfree}]
First, suppose that $G$ is edge-based. Assume for the sake of a contradiction that $G$ contains $K_4$ as a minor. Then, by Lemma \ref{lem:k4minor_top}, $G$ contains $K_4$ also as a topological minor. This implies that $G$ contains a subgraph that is a subdivision of $K_4$. During the leaf-shrinking procedure, it might be possible to suppress all vertices that have degree 2 in this subdivision. However, the resulting subgraph $K_4$ cannot be further reduced and hence $G$ cannot be edge-based; a contradiction.

Now, assume that $G$ is $K_4$-minor free. Then, by Lemma \ref{lem:k4minor_top}, $G$ does not contain $K_4$ as a topological minor. Moreover, by assumption $G$ is connected. Thus, by Proposition \ref{prop:K4_GSP}, $G$ is a GSP graph, which by Lemma \ref{lem:GSP_edgebased} implies that $G$ is edge-based. This completes the proof.
\end{proof}

\begin{Rem}\label{Rem_subgraph_edgebased}
An immediate consequence of Proposition~\ref{prop:K4-minorfree} is that if $G$ is an edge-based graph and $H$ is a connected subgraph of $G$ with at least two vertices, then $H$ is edge-based, too.
\end{Rem}

We proceed by recalling a statement from \cite{Diestel2017} concerning the maximum number of edges a $K_4$-minor free graph can have.
\begin{Lem}[{adapted from \cite[Corollary 7.3.2]{Diestel2017}}] \label{lem:K4free_maxedges}
Every edge-maximal graph $G=(V,E)$ without $K_4$ as a minor has $2|V|-3$ edges.
\end{Lem}

We now state a sufficient property for edge-basedness.
\begin{Prop}[{adapted from \cite[Theorem 4.8]{Hamann2021}}]\label{prop:tier_edgebased}
Let $G=(V,E)$ be a connected graph with $|V| \geq 2$ and tier$(G)\leq 2$. Then, $G$ is edge-based.
\end{Prop}

\begin{proof}
Let $G=(V,E)$ be a connected graph with $|V|\geq 2$ and tier$(G) \leq 2$. Assume for the sake of a contradiction that $G$ is not edge-based. Then, by Proposition \ref{prop:K4-minorfree} and Lemma \ref{lem:k4minor_top}, $G$ contains $K_4$ as a topological minor. It is now easily checked that at least 3 edges need to be removed from this subdivision of $K_4$ to obtain a tree. Thus, tier$(G) \geq 3$; a contradiction. This completes the proof.
\end{proof}

Next, we want to show that every edge-based phylogenetic network is automatically proper; an observation that was already stated in \cite[Theorem 4.14]{Hamann2021}. In order to prove this statement, we need the following theorem concerning GSP graphs, which basically implies that GSP graphs can be reduced to any one of their edges by operations of types (i), (ii), and (iii).

\begin{Theo}[{\cite[Theorem 4.1]{Ho1999}}]
\label{thm_edgechooseGSP} Let $G$ be a GSP graph. Then, for any edge $e = \{u, v\}$ of $G$, $G$ is a GSP graph with terminals $u$ and $v$.
\end{Theo}

We are now in a position to prove that every edge-based network is automatically proper.

\begin{Cor}\label{cor_eb=proper} Let $N$ be an edge-based phylogenetic network. Then, $N$ is proper.
\end{Cor}

\begin{proof}
First, suppose that there is a cut vertex, $u$ say, whose removal disconnects $N$ in a way that one remaining component $\mathcal{C}$ contains no leaf of $N$. Note that this in particular implies that $N\neq K_2$. We now re-introduce $u$ to $\mathcal{C}$ and consider $\mathcal{C}$ a bit more in-depth. As $\mathcal{C}$ contains no leaf of $N$, all vertices in $\mathcal{C}$  (possibly except for $u$) have degree at least three (as $N$ has no degree-2 vertices and as $\mathcal{C}$ has no leaves other than possibly $u$). Note that $\mathcal{C}$ is connected and contains at least two vertices, so $u$ has at least one neighbor $w$. However, as $N$ is edge-based, by Remark \ref{Rem_subgraph_edgebased}, so is $\mathcal{C}$, and by Lemma \ref{lem:GSP_edgebased}, $\mathcal{C}$ is a GSP graph, which, by Theorem \ref{thm_edgechooseGSP}, we can reduce to edge $\{u,w\}$. In order to do so, $u$ cannot be deleted (even if it is a leaf) or suppressed (even if it has degree 2), and as $\mathcal{C}$ has no parallel edges or loops (as $N$ is a phylogenetic network), the reduction must start with suppressing another degree-2 vertex -- but this contradicts the fact that all vertices in $\mathcal{C}$ other than possibly $u$ have degree at least three. So this is not possible.

Similarly, if there is a cut edge $e=\{u,v\}$ whose removal disconnects $N$ in a way that one remaining component $\mathcal{C}$, say the one containing vertex $u$, contains no leaf of $N$, we can repeat the same argument (with the exception that $u$ cannot be a leaf in $\mathcal{C}$) to derive a contradiction. 
This completes the proof.
\end{proof}

Finally, we state a result concerning the decomposition of edge-based networks.
\begin{Prop}\label{prop:blobs_blocks_edgebased}
Let $N$ be a phylogenetic network on $X$ with $|X| \geq 2$. Then, the following are equivalent:
\begin{enumerate}[(i)]
    \item $N$ is edge-based;
    \item every non-trivial blob of $N$ is edge-based;
    \item every block of $N$ is edge-based.
\end{enumerate}
\end{Prop}

\begin{proof}
A proof for the equivalence of (i) and (ii) is given in \cite[Proposition 1]{Fischer2020}, where \enquote{(i) $\Rightarrow$ (ii)} is shown by way of contradiction, and \enquote{(ii) $\Rightarrow$ (i)} is shown by induction on the number of non-trivial blobs.

We now show that (i) implies (iii). Therefore, let $N$ be edge-based. By Proposition \ref{prop:K4-minorfree}, $N$ is $K_4$-minor free. Thus, clearly all subgraphs of $N$, and therefore in particular all blocks of $N$, are $K_4$-minor free (cf. Remark \ref{Rem_subgraph_edgebased}) and hence edge-based by Proposition \ref{prop:K4-minorfree}.

Finally, to see that (iii) implies (i), suppose that every block of $N$ is edge-based. As every block of $N$ is loopless (since $N$ is loopless), every such block is a GSP graph by Lemma \ref{lem:GSP_edgebased}. Moreover, as every block of a GSP graph is an SP graph (Lemma \ref{lem:GSP_blocks_SP}), each block of $N$ is an SP graph. Again using Lemma \ref{lem:GSP_blocks_SP}, this implies that $N$ is a GSP graph and thus, by Lemma \ref{lem:GSP_edgebased}, $N$ is edge-based. This completes the proof.
\end{proof}

\subsubsection{Computational complexity of vertex-deletion, edge-deletion, and edge-contraction problems}
\label{sec:computational_complexity}
In this section, we recall some results on the computational complexity of so-called vertex-deletion, edge-deletion, and edge-contraction problems.

Suppose $\pi$ is a property on graphs. Then, the corresponding \emph{vertex-deletion problem} is the following: Given a graph $G$, find a set of vertices of minimum cardinality whose deletion results in a graph satisfying property $\pi$. Note that an equivalent formulation of this problem is the \emph{maximum subgraph} problem: Given a graph $G$, find an induced subgraph of $G$ of maximum order that satisfies $\pi$. If this induced subgraph is additionally required to be connected, the problem is called the \emph{connected maximum subgraph} (or \emph{connected vertex-deletion}) problem \cite{Yannakakis1979}.

Similarly, the \emph{edge-deletion} (\emph{edge-contraction}) problem is to find a set of edges of $G$ of minimum cardinality whose deletion (contraction) results in a graph satisfying property $\pi$.

\paragraph{Vertex-deletion problems.}
In order to state a result from \cite{Yannakakis1979} on the connected vertex-deletion problem, we require the following definitions (taken from \cite{Yannakakis1979}). We first remark that the input graphs for the connected vertex-deletion problem considered by \cite{Yannakakis1979} are connected simple graphs (personal communication). Now, a graph property $\pi$ is called \emph{non-trivial} (on some domain $D$ of graphs) if it is true for some graph but not for all graphs in $D$. Moreover, $\pi$ is called \emph{interesting} if there are arbitrarily large graphs in $D$ satisfying $\pi$. Finally, $\pi$ is called \emph{hereditary on induced subgraphs} if, whenever $G$ is a graph satisfying $\pi$, then the deletion of any vertex does not result in a graph violating $\pi$.
Based on this, we have:
\begin{Theo}[{\cite[Theorem 1]{Yannakakis1979}}] \label{theo:vertex_deletion_NP}
The connected maximum subgraph problem for graph properties that are hereditary on induced subgraphs, and non-trivial and interesting on connected graphs, is NP-hard.
\end{Theo}

\paragraph{Edge-deletion and edge-contraction problems.}
In order to discuss edge-deletion and edge-contraction problems, we recall some definitions and results from \cite{Asano1983} and \cite{ElMallah1988}.

We begin by considering edge-deletion problems for graph properties characterizable by a set $\mathbf{F}$ of forbidden graphs. Adapting notation from \cite{ElMallah1988}, we say that a graph $G$ is $\mathbf{F}$-minor free if $G$ does not contain a minor isomorphic to any member of $\mathbf{F}$. Moreover, we use $P_{\text{ED}}(\mathbf{F})$ to denote the edge-deletion problem corresponding to a class of graphs in which each member is an $\mathbf{F}$-minor free graph. In other words, given an arbitrary graph $G$ and a set of forbidden graphs $\mathbf{F}$, $P_{\text{ED}}(\mathbf{F})$ is the problem of finding the minimum number of edges of $G$ whose deletion results in a subgraph $G'$ such that $G'$ is $\mathbf{F}$-minor free. Now, \cite{ElMallah1988} obtained the following result:

\begin{Theo}[{adapted from \cite[Theorem 1]{ElMallah1988}}] \label{Theo_Hardness_EdgeDeletion}
Let $\mathbf{F}$ be a set of graphs in which each member is a simple biconnected graph of minimum degree at least three. Then, the edge-deletion problem $P_{\text{ED}}(\mathbf{F})$ is NP-hard.
\end{Theo}

We now consider edge-contraction problems as studied by \cite{Asano1983}. First, let $G$ be a multigraph. Then, the \emph{simple graph} of $G$ is obtained by replacing every multiple edge of $G$ with a single edge and deleting all loops of $G$. Moreover, if $\pi$ is a graph property, then it is called \emph{hereditary on contractions} if, for any graph $G$ satisfying $\pi$, all contractions of $G$ also satisfy $\pi$. Moreover, $\pi$ is called \emph{non-trivial on connected graphs} if it is true for infinitely many connected graphs and false for infinitely many graphs. Furthermore, a property $\pi$ is \emph{determined by the simple graph} if, for any graph $G$, $G$ satisfies $\pi$ if and only if its underlying simple graph satisfies $\pi$. Finally, $\pi$ is \emph{determined by the biconnected components} if, for any graph $G$, $G$ satisfies $\pi$ if and only if all biconnected components of $G$ satisfy $\pi$. Now, let $P_{\text{EC}}(\pi)$ denote the edge-contraction problem of, given any graph $G$, finding a set of edges of minimum cardinality whose contractions results in a graph satisfying property $\pi$. Then, we have the following result from \cite{Asano1983}.

\begin{Theo}[{adapted from \cite{Asano1983}}] \label{Theo_Hardness_EdgeContraction}
The edge-contraction problem $P_{\text{EC}}(\pi)$ is NP-hard for a property $\pi$ satisfying the following four conditions:
    \begin{enumerate}[(C1)]
        \item $\pi$ is non-trivial on connected graphs;
        \item $\pi$ is hereditary on contractions;
        \item $\pi$ is determined by the simple graph; and
        \item $\pi$ is determined by the biconnected components.
    \end{enumerate}
\end{Theo}

Having recalled all relevant notation and known results, we are now in a position to turn our attention to the main aspect of the present manuscript, namely the introduction and analysis of proximity measures for edge-basedness.

\section{Results}
\subsection{Proximity measures for edge-basedness based on phylogenetic networks} \label{sec_nwbased}
As not all phylogenetic networks are edge-based, we now introduce the first four measures that can be used to assess the proximity of a phylogenetic network to being edge-based. The measures presented in this section are all based on the given network itself, whereas the measures in the following section will be based on the network's leaf shrink graph.

\begin{Def} \label{def:proximity_networkbased}
Let $N=(V,E)$ be a phylogenetic network on $X$ with $|X| \geq 2$. 
    \begin{enumerate}[(1)]
        \item Let $d^{ED}(N) \coloneqq \min\{k \,\vert\, G' = (V,E') \text{ with } E' \subseteq E, |E'| = |E|-k \text{, $G'$ edge-based}\}$.
        \item Let $d^{ER}(N)\coloneqq\min \{k \,\vert\, G'=(V,E') \text{ with } |E'| = |E| \text{ and } |E \cap E'|=|E|-k \text{, $G'$ edge-based}\}$.
        \item Let $d^{EC}(N) \coloneqq \min\{k \,\vert\, k = \text{number of edges that are contracted in $N$ to obtain $G'$, $G'$ edge-based}\}$.
       \item Let $d^{VD}(N) \coloneqq \min \{k \,\vert\, G' = N[V'] \text{ with } V'\subseteq V, |V'|=|V|-k \text{, $G'$ edge-based}\}$.
    \end{enumerate}
\end{Def}

In words, (1) is the minimum number of edges of $N$ that need to be deleted in order to obtain an edge-based graph $G'$. Similarly, (2) is the minimum number of edges of $N$ which need to be relocated to obtain an edge-based graph $G'$, and (3) is the minimum number of edges of $N$ that need to be contracted to obtain an edge-based graph $G'$. For (3), when we contract edges, we do not introduce multiple edges or loops, i.e., we keep the graph simple (note that deleting loops and copies of parallel edges are valid operations of the leaf-shrinking procedure and thus the proximity measure is not affected by this convention). Finally, (4) is the minimum number of vertices of $N$ that need to be deleted to obtain an edge-based graph $G'$, i.e., $G'$ is an induced subgraph of $N$ of maximum order that is edge-based.\\

In any case, we clearly have $d^{ED}(N) = d^{ER}(N) = d^{EC}(N) = d^{VD}(N)=0$ if and only if $N$ is edge-based. If $N$ is not edge-based, all four measures are strictly positive.\\

We remark that for technical reasons, we sometimes apply these proximity measures to general connected graphs (e.g., in the proof of Proposition \ref{N:dElowerbound}), for which they are defined analogously.\\

Before we can analyze the introduced proximity measures more in-depth, note that all of them  measure the distance from $N$ to an edge-based graph (and not necessarily to an edge-based phylogenetic network) as the operations used (edge deletions, edge relocations, edge contractions, and vertex deletions) in some cases inevitably lead to graphs that violate the definition of a phylogenetic network. For instance, the resulting edge-based graphs may contain degree-2 vertices, parallel edges, or loops. An example is depicted in Figure \ref{Fig_LSNoPhylo}. Here, it suffices to delete one edge of $N$ to make it edge-based, but the resulting graph inevitably contains a degree-2 vertex. In order to stay in the space of phylogenetic networks, we could continue to modify the graph by applying the leaf-shrinking procedure. However, it is easily seen that this will lead to a single edge in this example. Thus, it is not always possible to obtain a phylogenetic network distinct from $K_2$ this way. Alternatively, if $N$ is a non-edge-based phylogenetic network and $G$ is its closest edge-based graph (according to one of the proximity measures introduced above), we can simply turn $G$ into a phylogenetic network $N'$ by attaching additional leaves to all degree-2 vertices, parallel edges, or loops that $G$ may contain. Clearly, $N'$ will still be edge-based. Thus, it is possible to measure the distance from $N$ to an edge-based phylogenetic network $N'$, i.e., it is possible to stay in the space of phylogenetic networks. However, for simplicity, we measure the distance from $N$ to an edge-based graph in the following.

\begin{figure}[htbp]
    \centering
    \includegraphics[scale=0.3]{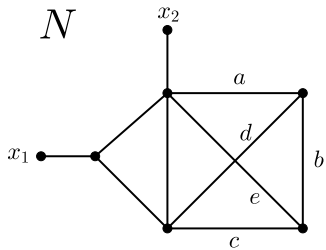}
    \caption{Non-edge-based phylogenetic network $N$ with $d^{ED}(N)=1$. It suffices to delete \emph{one} of the edges $a, b, c, d$, or $e$ to make $N$ edge-based, but the resulting graph will contain a degree-2 vertex and is thus no longer a phylogenetic network. The same applies if more than one edge gets deleted.}
    \label{Fig_LSNoPhylo}
\end{figure}

We conclude this section with the following remark concerning the edge deletion and replacement proximity measures.

\begin{Rem} \label{rem_dER_nocutedge} It can be easily seen that concerning $d^{ED}$ and $d^{ER}$, due to the minimization, no cut edge gets ever deleted, respectively replaced. This is due to the fact that we want to reach an edge-based graph $G'$, so in particular a connected graph, and that each $K_4$ subdivision of $N$ (if any) must be contained in a blob of $N$, as cut edges are not contained in any cycle. Thus, deleting, respectively moving a cut edge will never have any impact of the $K_4$ minors of $N$, which shows that such moves are never necessary or helpful in any way (deleting a cut edge is even harmful since it destroys connectivity) in order to reach an edge-based graph. We will use this insight later on.
\end{Rem}

\subsubsection{Relationships among the network-based proximity measures for edge-basedness}
\noindent In the following, we analyze the relationships among the four different network-based proximity measures for edge-basedness. We begin by showing that $d^{ED}(N) = d^{ER}(N)$.

\begin{Theo}\label{N:E=U}
Let $N$ be a phylogenetic network on $X$ with $|X| \geq 2$. Then, $d^{ED}(N) = d^{ER}(N)$.
\end{Theo}

\begin{proof}
We first show that $d^{ER}(N) \leq d^{ED}(N)$. Suppose that $d^{ED}(N)=k$, i.e., $N$ contains a set of $k$ non-cut edges whose deletion leads to an edge-based graph $G'$. We now argue that a suitable relocation of these $k$ edges also leads to an edge-based graph. To see this, note that we can iteratively relocate these $k$ non-cut edges such that they become parallel edges. Let $\widehat{G}$ denote the graph resulting from this procedure. Clearly, the $k$ parallel edges of $\widehat{G}$ can be deleted during the leaf-shrinking procedure, resulting in the graph $G'$, which by assumption is edge-based. Thus, $\widehat{G}$ is edge-based, and we have $d^{ER}(N) \leq d^{ED}(N)$ as claimed.

Now, we show that $d^{ED}(N) \leq d^{ER}(N)$. Suppose that $d^{ER}(N)=k$, i.e., $N$ contains a set of $k$ edges whose relocation leads to an edge-based graph $G'$. By Remark \ref{rem_dER_nocutedge}, none of these edges is a cut edge. Now, consider the connected graph $\widehat{G}$ obtained from deleting these $k$ edges from $N$ instead of relocating them. Then, $\widehat{G}$ is a subgraph of $G'$. As $G'$ is edge-based by assumption, by Remark \ref{Rem_subgraph_edgebased} we conclude that $\widehat{G}$ is edge-based, too (note that as $|X| \geq 2$, $\widehat{G}$ contains at least two vertices).
Thus, $\widehat{G}$ is edge-based, and we have $d^{ED}(N) \leq d^{ER}(N)$ as claimed.

In summary, $d^{ER}(N) = d^{ED}(N)$, which completes the proof.
\end{proof}

While, $d^{ED}(N)=d^{ER}(N)$, there is no direct relationship among the other proximity measures. In particular, we have the following:
\begin{itemize}
    \item There exist non-edge-based phylogenetic networks $N$ such that $d^{ED}(N)=d^{ER}(N)=d^{EC}(N)=d^{VD}(N)$. As an example, for the phylogenetic network $N_2$ depicted in Figure \ref{Fig_3Networks}, it is easily verified that $d^{ED}(N)=d^{ER}(N)=d^{EC}(N)=d^{VD}(N)=1$.
    
    \item There exist non-edge-based phylogenetic networks $N$ such that $d^{EC}(N) < d^{ED}(N)=d^{ER}(N)$. An example is depicted in Figure \ref{Fig_N_contract_vs_delete}(i), where $d^{EC}(N)=2$, whereas $d^{ED}(N)=d^{ER}(N)=3$.
    
    \item There exist non-edge-based phylogenetic networks $N$ such that $d^{EC}(N) > d^{ED}(N)=d^{ER}(N)$. An example is depicted in Figure \ref{Fig_N_contract_vs_delete}(ii), where $d^{EC}(N)=2$, whereas $d^{ED}(N)=d^{ER}(N)=1$.
    
    \item There exist non-edge-based phylogenetic networks $N$ such that $d^{VD}(N) < d^{EC}(N)=d^{ED}(N)=d^{ER}(N)$. An example is depicted in Figure \ref{Fig_N_vertexdelete_vs_rest}(i), where $d^{VD}(N)=1$, whereas $d^{EC}(N)=d^{ED}(N)=d^{ER}(N)=2$.
    
    \item There exist non-edge-based phylogenetic networks $N$ such that $d^{VD}(N) > d^{EC}(N)=d^{ED}(N)=d^{ER}(N)$. An example is depicted in Figure \ref{Fig_N_vertexdelete_vs_rest}(ii), where $d^{VD}(N)=8$, whereas $d^{EC}(N)=d^{ED}(N)=d^{ER}(N)=5$.
\end{itemize}

\begin{figure}[htbp]
    \centering
    \includegraphics[scale=0.35]{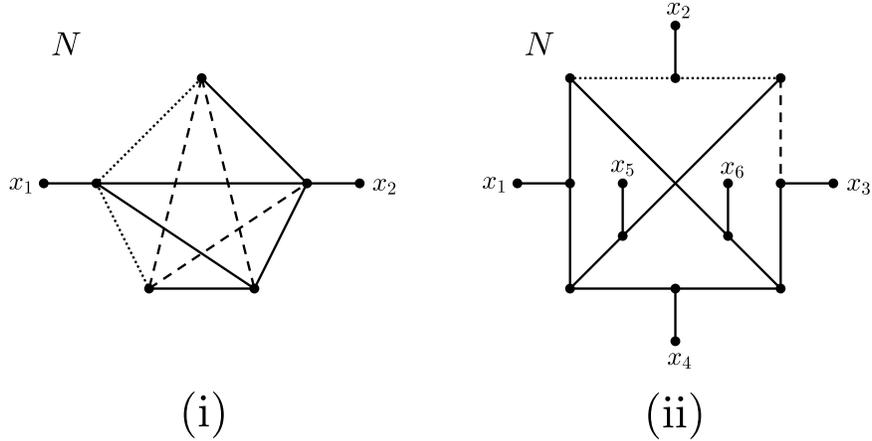}
    \caption{(i) Phylogenetic network $N$ with $d^{EC}(N)=2$ and $d^{ED}(N)=d^{ER}(N)=3$. To see that $d^{EC}(N)=2$, note that contracting for instance the two dotted edges yields an edge-based graph, whereas contracting only one edge of $N$ yields a graph containing $K_4$ as a minor. Similarly, to see that $d^{ED}(N)=d^{ER}(N)=3$, it is easily checked that deleting/relocating for instance the three dashed edges yields an edge-based graph, whereas deleting/relocating strictly fewer than three edges of $N$ yields a graph containing $K_4$ as a minor.
    (ii) Phylogenetic network $N$ with $d^{EC}(N)=2$ and $d^{ED}(N)=d^{ER}(N)=1$. Again, it is easily checked that contracting for instance the two dotted edges, respectively deleting/relocating for instance the dashed edge of $N$ yields an edge-based graph, whereas contracting, respectively deleting/relocating, strictly fewer edges results in a graph containing $K_4$ as a minor.}
    \label{Fig_N_contract_vs_delete}
\end{figure}

\begin{figure}[htbp]
    \centering
    \includegraphics[scale=0.25]{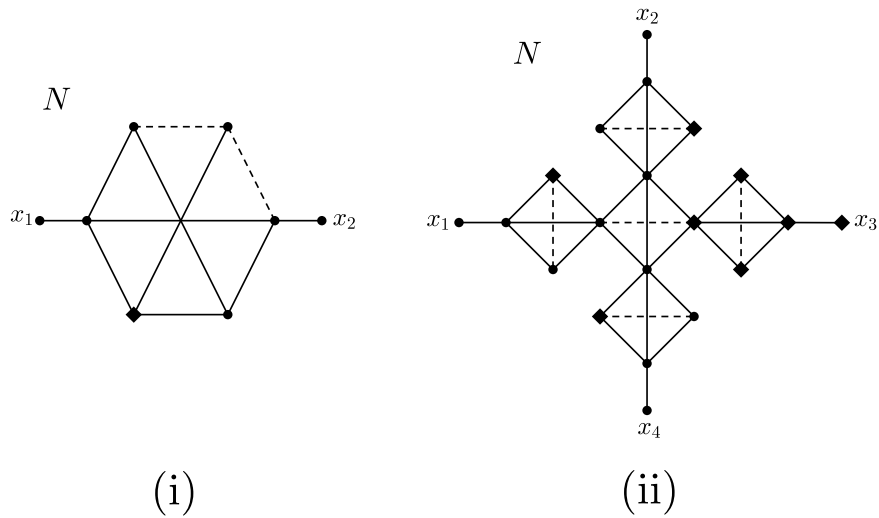}
    \caption{(i) Phylogenetic network $N$ with $d^{VD}(N)=1$ and $d^{EC}(N)=d^{ED}(N)=d^{ER}(N)=2$. It is easily seen that at least one vertex needs to be deleted from $N$ to obtain an edge-based graph and a possible choice is the vertex depicted as a diamond. Similarly, it can easily be verified that at least two edges need to be deleted/relocated/contracted to obtain an edge-based graph and a possible choice are the two dashed edges.
    (ii) Phylogenetic network $N$ with $d^{VD}(N)=8$ and $d^{EC}(N)=d^{ED}(N)=d^{ER}(N)=5$. Here, $d^{VD}(N)=8$, because in order to obtain an edge-based graph, at least one vertex in each of the 5 induced $K_4$'s needs to be deleted; however, in order to obtain a connected graph, a vertex of the \enquote{central} $K_4$ can only be deleted, if one of the non-trivial blocks bordering it is completely deleted. Thus, a total of 8 vertices needs to be deleted to obtain a connected $K_4$-minor free and thus edge-based graph. A possible choice is given by the vertices depicted as diamonds. On the other hand, $d^{ED}(N)=d^{ER}(N)=d^{EC}(N)=5$, because one edge out of each induced $K_4$ needs to be deleted/relocated/contracted to obtain a connected $K_4$-minor free and thus edge-based graph. A possible choice is given by the dashed edges.}
    \label{Fig_N_vertexdelete_vs_rest}
\end{figure}

\subsubsection{Upper and lower bounds for some of the network-based proximity measures for edge-basedness}
In this section, we derive upper and lower bounds for the proximity measures $d^{ED}=d^{ER}$.
We begin by stating an upper bound for $d^{ED}(N)=d^{ER}(N)$ based on the tier of $N$ that follows from Proposition \ref{prop:tier_edgebased}.

\begin{Cor}\label{N:dEleqtierG}
Let $N$ be a phylogenetic network on $X$ with $|X| \geq 2$. Then, $d^{ER}(N)=d^{ED}(N) \leq \text{tier}(N)$. Moreover, $d^{ER}(N)=d^{ED}(N)=\text{tier}(N)=0$ if $N$ is a phylogenetic tree, and $d^{ER}(N)=d^{ED}(N) \leq \text{tier}(N)-2$ if $N$ is not edge-based.
\end{Cor}

\begin{proof}
First, suppose that $N$ is an edge-based phylogenetic network. Then, $d^{ED}(N)=0$. Moreover, by definition, $\text{tier}(N) \geq 0$, and so $d^{ED}(N) \leq \text{tier}(N)$ if $N$ is edge-based. In the special case that $N$ is a phylogenetic tree, $d^{ED}(N)=\text{tier}(N)=0$ (because trees are trivially edge-based and the tier of a tree is zero).

Now, suppose that $N$ is not edge-based. Let $k = \text{tier}(N)$. By Proposition \ref{prop:tier_edgebased}, we have $k > 2$ (as otherwise, $N$ would be edge-based). Let $G'$ denote the tree obtained from $N$ by deleting $k$ suitable edges. If we now re-introduce two of the $k$ edges, we obtain a graph $\widehat{G}$ with tier$(\widehat{G})=2$. By Proposition \ref{prop:tier_edgebased}, $\widehat{G}$ is edge-based. In particular, this implies that we can turn $N$ into an edge-based graph by deleting at most $k-2$ edges. Hence, $d^{ED}(N) \leq \text{tier}(N)-2$. This completes the proof for $d^{ED}(N)$, and the same statements for $d^{ER}(N)$ follow by Theorem \ref{N:E=U}.
\end{proof}

We now derive a lower bound for $d^{ED}(N)=d^{ER}(N)$.

\begin{Prop}\label{N:dElowerbound}
Let $N=(V(N),E(N))$ be a phylogenetic network on $X$ with $|X| \geq 2$. Then,
$d^{ER}(N)= d^{ED}(N) \geq \max\{0, |E(N)| - 2|V(N)| + |X| + 3\}$.
\end{Prop}

\begin{proof}
By Lemma \ref{lem:K4free_maxedges}, every edge-maximal $K_4$-minor free graph $G=(V(G),E(G))$ has $2|V(G)|-3$ edges. Thus, if $N$ is not edge-based, at least $|E(N)|-(2|V(N)| -3)$ ges need to be deleted to obtain an edge-based graph. In particular, $d^{ED}(N) \geq \max\{0, |E(N)|-(2|V(N)| -3)\}$ (where $d^{ED}(N)=0$ if $N$ is edge-based). However, it is easily seen that this bound can be improved by noting that we can obtain a graph, $G=(V(G),E(G))$ say, from $N$ by deleting all elements of $X$ together with their incident edges such that $d^{ED}(N)=d^{ED}(G)$ (note that the equality is due to the fact that cut edges never get deleted when turning a non-edge-based graph into an edge-based one (Remark \ref{rem_dER_nocutedge}), and thus $d^{ED}(N)$ is not determined by their number). Now, $|V(G)|=|V(N)|-|X|$ and $|E(G)|=|E(N)|-|X|$, and again using Lemma \ref{lem:K4free_maxedges}, we obtain
\begin{align*}
    d^{ED}(N) = d^{ED}(G) &\geq \max\{0, |E(N)|-|X| - (2 (|V(N)|-|X|)-3) \} \\
    &= \max\{0, |E(N)| - 2 |V(N)|+ |X| + 3\}.
\end{align*}
The same statement for $d^{ER}(N)$ follows by Theorem \ref{N:E=U}, which completes the proof.
\end{proof}

Next, we will turn our attention to a different class of proximity measures.

\subsection{Proximity measures for edge-basedness based on leaf shrink graphs}
While all proximity measures for edge-basedness introduced so far were based on the underlying phylogenetic network, we now introduce analogous measures based on the leaf shrink graph ($\LS$ graph-based, for short) of the network.

\begin{Def} \label{def:proximity_LSbased}
Let $N = (V,E)$ be a phylogenetic network on $X$ with $|X| \geq 2$ and let $\LS(N)=(V_{\LS},E_{\LS})$ denote its leaf shrink graph.

\begin{enumerate}[(1)]
    \item Let $d_{ED}(N) \coloneqq \min\{k \,\vert\, G' = (V_{\LS},E') \text{ with } E' \subseteq E_{\LS}, |E'| = |E_{\LS}|-k \text{, $G'$ edge-based}\}$.
    \item Let $d_{ER}(N)\coloneqq\min \{k \,\vert\, G'=(V_{\LS},E') \text{ with } |E'| = |E_{\LS}| \text{ and } |E_{\LS} \cap E'|=|E_{\LS}|-k \text{, $G'$ edge-based}\}$.
    \item Let $d_{EC}(N) \coloneqq \min\{k \,\vert\, k = \text{number of edges that are contracted in $\LS(G)$ to obtain $G'$, $G'$ edge-based}\}$.
    \item Let $d_{VD}(N) \coloneqq \min \{k \,\vert\, G'=\LS(N)[V'] \text{ with } V' \subseteq V_\LS, |V'|=|V_{\LS}|-k \text{, $G'$ edge-based}\}$.
\end{enumerate}
\end{Def}

Analogously to the network-based proximity measures for edge-basedness introduced in Definition \ref{def:proximity_networkbased}, (1) and (2) refer to the minimum number of edges of $\LS(N)$ that need to be deleted, respectively relocated, to obtain an edge-based graph, and by the same arguments used in Section \ref{sec_nwbased}, these edges are no cut edges (cf. Remark \ref{rem_dER_nocutedge}).
Similarly, (3) is the minimum number of edges of $\LS(N)$ that need to be contracted to obtain an edge-based graph (where we again keep the graph simple, i.e., where we do not introduce parallel edges or loops). Finally, (4) is the minimum number of vertices that need to be deleted from $\LS(N)$ to obtain an edge-based graph, i.e., an induced subgraph of $\LS(N)$ of maximum order that is edge-based.\\

We clearly have $d_{ED}(N)=d_{ER}(N)=d_{EC}(N)=d_{VD}(N)=0$ if and only if $N$ is edge-based. If $N$ is not edge-based, all four measures are strictly positive.\\

In the following, we will explore the relationships among the different $\LS$ graph-based proximity measures, before relating them to the network-based proximity measures introduced in the previous section.

\subsubsection{Relationships among the \texorpdfstring{$\LS$}{LS} graph-based proximity measures for edge-basedness}
Recall that for the network-based proximity measures, we obtained the identity $d^{ED}(N)=d^{ER}(N)$ (Theorem \ref{N:E=U}). It immediately follows from the proof of this theorem that the same identity holds for the corresponding $\LS$ graph-based proximity measures.

\begin{Cor}\label{LS:E=U}
Let $N$ be a phylogenetic network on $X$ with $|X| \geq 2$. Then, $d_{ED}(N) = d_{ER}(N)$.
\end{Cor}

\begin{proof}
The proof is analogous to the proof of Theorem \ref{N:E=U}; we simply repeat the argument for deleting, respectively relocating, edges for $\LS(N)$ (instead of $N$).
\end{proof}

We now show that under certain circumstances, we also have an equality of $d_{EC}(N)$ and $d_{VD}(N)$.

\begin{Cor} \label{cor:LS_dZ_dV_equality}
Let $N$ be a phylogenetic network on $X$ with $|X| \geq 2$ such that $\LS(N)=K_i$ with $i=2$ or $i \geq 4$. Then $d_{EC}(N) = d_{VD}(N) = 0$ if $i = 2$ and $d_{EC}(N) = d_{VD}(N) = i-3$ if $i \geq 4$.
\end{Cor}

\begin{proof}
First, assume that $\LS(N)=K_2$. Then, $N$ is edge-based and thus $d_{EC}(N)=d_{VD}(N)=0$.

Now, suppose that $\LS(N)=K_i$ with $i \geq 4$, which implies that $N$ is not edge-based. 
We first show that $d_{VD}(N)=i-3$. As $N$ is not edge-based, we begin by deleting one vertex in $\LS(N) = K_i$, resulting in the complete graph $K_{i-1}$. If $K_{i-1}$ is not edge-based, we delete another vertex and obtain $K_{i-2}$. In particular, in order to obtain an edge-based (and thus $K_4$-minor free) graph, we have to delete as many vertices as we need to obtain $K_3$ from $\LS(N) = K_i$. These are $i-3$ many and thus $d_{VD}(N) = i-3$.
Recalling the convention that the contraction of edges does not lead to loops or parallel edges, the proof that $d_{EC}(N) = i-3$ is completely analogous to the proof of $d_{VD}(N)=i-3$.

This completes the proof.
\end{proof}

Apart from the identities stated in Corollaries \ref{LS:E=U} and \ref{cor:LS_dZ_dV_equality}, there is no direct relationship among the four $\LS$ graph-based proximity measures (analogous to what we have seen for the network-based proximity measures). In particular, we have:
\begin{itemize}
    \item There exist non-edge-based phylogenetic networks $N$ such that $d_{ED}(N)=d_{ER}(N)=d_{EC}(N)=d_{VD}(N)$. As an example, consider the phylogenetic network $N_2$ depicted in Figure \ref{Fig_3Networks}. It is easily checked that $\LS(N_2)=K_4$ and $d_{ED}(N)=d_{ER}(N)=d_{EC}(N)=d_{VD}(N)=1$.
    \item There exist non-edge-based phylogenetic networks $N$ such that $d_{EC}(N) < d_{ED}(N)=d_{ER}(N)$. As an example, consider $\LS(N)=K_5$ depicted in Figure \ref{Fig_LS_contract_vs_delete}(i), where $d_{EC}(N)=2$, whereas $d_{ED}(N)=d_{ER}(N)=3$.
    \item There exist non-edge-based phylogenetic networks $N$ such that $d_{EC}(N) > d_{ED}(N)=d_{ER}(N)$. As an example, consider $\LS(N)$ depicted in Figure \ref{Fig_LS_contract_vs_delete}(ii), where $d_{EC}(N)=4$, whereas $d_{ED}(N)=d_{ER}(N)=2$.
    \item There exist non-edge-based phylogenetic networks $N$ such that $d_{VD}(N) < d_{EC}(N)=d_{ED}(N)=d_{ER}(N)$. As an example, consider $\LS(N)$ depicted in Figure \ref{Fig_LS_vertexdelete_vs_rest}(i), where $d_{VD}(N)=1$, whereas $d_{EC}(N)=d_{ED}(N)=d_{ER}(N)=2$.
    \item There exist non-edge-based phylogenetic networks $N$ such that $d_{VD}(N) > d_{EC}(N)=d_{ED}(N)=d_{ER}(N)$. As an example, consider $\LS(N)$ depicted in Figure \ref{Fig_LS_vertexdelete_vs_rest}(ii), where $d_{VD}(N)=7$, whereas $d_{EC}(N)=d_{ED}(N)=d_{ER}(N)=5$.
\end{itemize}

\begin{figure}[htbp]
    \centering
    \includegraphics[scale=0.35]{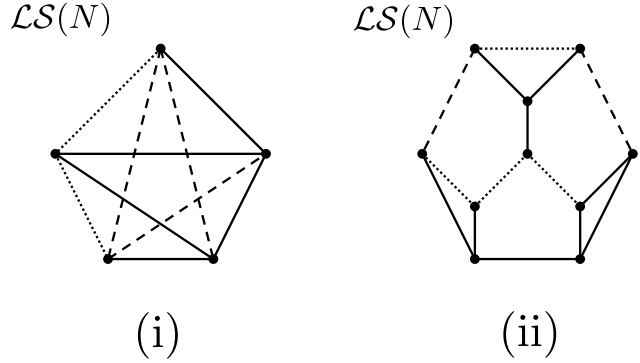}
    \caption{(i) Leaf shrink graph $\LS(N)=K_5$ yielding $d_{EC}(N)=2$ and $d_{ED}(N)=d_{ER}(N)=3$. It is easily checked that at least two edges need to be contracted (three edges need to be deleted/relocated) to obtain an edge-based graph from $\LS(N)$, and a possible choice is given by the two dotted (three dashed) edges. (ii) Leaf shrink graph $\LS(N)$ yielding $d_{EC}(N)=4$ and $d_{ED}(N)=d_{ER}(N)=2$. It is easily checked that at least two edges need to be deleted/relocated to obtain an edge-based graph from $\LS(N)$, and a possible choice is given by the two dashed edges. For, $d_{EC}(N)$, using the computer algebra system Mathematica \cite{Mathematica}, we exhaustively verified that $\LS(N)$ does not contain any subset of up to three edges whose contraction yields an edge-based graph. However, it suffices to contract for instance the four dotted edges.}
    \label{Fig_LS_contract_vs_delete}
\end{figure}

\begin{figure}[htbp]
    \centering
    \includegraphics[scale=0.35]{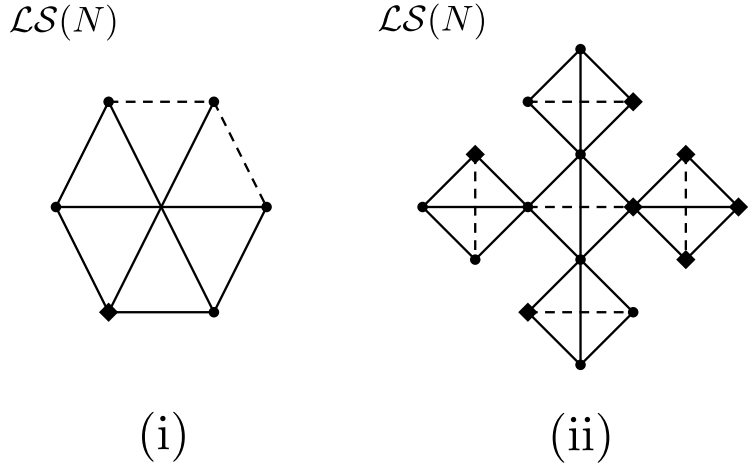}
    \caption{(i) Leaf shrink graph $\LS(N)$ yielding $d_{VD}(N)=1$ and $d_{EC}(N)= d_{ED}(N)=d_{ER}(N)=2$. It is easily checked that at least one vertex needs to be deleted from $\LS(N)$ to obtain an edge-based graph and a possible choice is the vertex depicted as a diamond. Similarly, it can easily be verified that at least two edges need to be contracted/deleted/relocated to obtain an edge-based graph and a possible choice is given by the two dashed edges. (ii) Leaf shrink graph $\LS(N)$ yielding $d_{VD}(N)=7$ and $d_{EC}(N)= d_{ED}(N)=d_{ER}(N)=5$. With the same reasoning as for the network depicted in Figure \ref{Fig_N_vertexdelete_vs_rest}(ii), at least one vertex in each of the five induced $K_4$'s needs to be deleted; however, in order to obtain a connected graph, this means that one of the \enquote{outer} $K_4$'s needs to be removed completely, yielding $d_{VD}(N)=7$ (a possible choice of seven vertices is given by the vertices depicted as diamonds). On the other hand, $d_{EC}(N)=d_{ED}(N)=d_{ER}(N)=5$, because one edge out of each of the five induced $K_4$'s needs to be contracted, deleted, or relocated, to obtain an edge-based graph. A possible choice is given by the five dashed edges.}
    \label{Fig_LS_vertexdelete_vs_rest}
\end{figure}

\subsubsection{Upper and lower bounds for some of the \texorpdfstring{$\LS$}{LS} graph-based proximity measures for edge-basedness}
For the network-based proximity measures for edge-basedness, we obtained lower and upper bounds for $d^{ED}(N) =d^{ER}(N)$. We now show that analogous bounds can be obtained for $d_{ED}(N)=d_{ER}(N)$.

\begin{Cor}\label{LS:dEleqtierLS}
Let $N$ be a phylogenetic network on $X$ with $|X| \geq 2$, and let $\LS(N)$ denote its leaf shrink graph. Then, $d_{ER}(N)=d_{ED}(N) \leq \text{tier}(\LS(N))$. In particular, $d_{ER}(N)=d_{ED}(N)=\text{tier}(\LS(N))=0$ if $N$ is edge-based, and $d_{ER}(N)=d_{ED}(N) \leq \text{tier}(\LS(N))-2$ if $N$ is not edge-based.
\end{Cor}

\begin{proof}
First, suppose that $N$ is edge-based. Then, $d_{ED}(N)=0$. Moreover, since $N$ is edge-based, we have $\LS(N)=K_2$ and thus clearly tier$(\LS(N))=0$.

Now, suppose that $N$ is not edge-based. Let $k = \text{tier}(\LS(N))$. Then, by Proposition \ref{prop:tier_edgebased}, we have $k > 2$. Let $G'$ denote the tree obtained from $\LS(N)$ by deleting $k$ suitable edges. If we now re-introduce two of these $k$ edges, we obtain a graph $\widehat{G}$ with tier$(\widehat{G})=2$, which is edge-based by Proposition \ref{prop:tier_edgebased}. In particular, we can obtain an edge-based graph from $\LS(N)$ by deleting at most $k-2$ edges. Thus, $d_{ED}(N) \leq \text{tier}(\LS(N))-2$. Using Corollary \ref{LS:E=U} to derive the same statements for $d_{ER}(N)$, this completes the proof.
\end{proof}

We now derive a lower bound for $d_{ED}(N)=d_{ER}(N)$.

\begin{Cor}\label{LS:dElowerbound}
Let $N$ be a phylogenetic network on $X$ with $|X| \geq 2$, and let $\LS(N)=(V_\LS, E_\LS)$ denote its leaf shrink graph. Then, $d_{ED}(N) = d_{ER}(N) \geq \max\{0, |E_\LS| - 2(|V_\LS|-3)\}$.
\end{Cor}

\begin{proof}
The statement is a direct consequence of the fact that every edge-maximal $K_4$-minor free graph $G=(V,E)$ has $2|V|-3$ edges (Lemma \ref{lem:K4free_maxedges}) and that $d_{ED}= d_{ER}$ is a non-negative function (see also proof of Proposition \ref{N:dElowerbound}).
\end{proof}

\subsection{Relationship between network-based and \texorpdfstring{$\LS$}{LS} graph-based proximity measures for edge-basedness}
In this section, we analyze how the network-based and the $\LS$ graph-based proximity measures relate to each other. We begin by showing that for all measures using edge-operations (edge deletion, edge relocation, and edge contraction), the  proximity measure based on the $\LS$ graph is a lower bound for the corresponding network-based proximity measure.

\begin{Prop}\label{prop:network_bound_for_LS}
Let $N$ be a phylogenetic network on $X$ with $|X| \geq 2$. Then, $d_\bullet (N) \leq d^\bullet (N)$ for $\bullet \in \{ED,ER,EC\}$.
\end{Prop}

\begin{proof} The crucial ingredients for this proof are the following three aspects: First, recall that by Lemma \ref{lem:k4minor_top}, for $K_4$ the concepts of minors and topological minors coincide. Second, by definition, $\LS(N)$ is a topological minor of $N$, which is why $N$ contains at least one subdivision of $\LS(N)$. We fix one such subdivision $\mathcal{S}$ of $\LS(N)$ in $N$. As a third step, note that this implies that every path in $\mathcal{S}$ corresponds to a unique edge in $\LS(N)$.

Now assume that we have a set of edges that need to be deleted/relocated/contracted in $\LS(N)$ in order to make this graph edge-based, i.e., $K_4$-minor free. In order to turn $N$ into a $K_4$-minor free graph, at least its subgraph $\mathcal{S}$, the subdivision of $\LS(N)$, needs to be made $K_4$-minor free, and so all operations applied to edges of $K_4$ need to be applied to their subdivided counterparts, i.e., their corresponding paths, in $\mathcal{S}$, too. For instance, if an edge from $\LS(N)$ needs to be deleted, at least the path corresponding to this edge in $\mathcal{S}$ needs to be cut by removing one edge. Similarly, if an edge needs to be contracted in $\LS(N)$, we need to contract at least one edge (but possibly more) in the respective path in $\mathcal{S}$. And if an edge from $\LS(N)$ needs to be relocated, then the subdivided version of this edge in $\mathcal{S}$ needs to be relocated by moving at least one edge from it, too, because otherwise the $K_4$ caused by this edge would still be present in $\mathcal{S}$ and thus also in $N$.

Thus, in summary, this shows $d_\bullet (N) \leq d^\bullet (N)$ for $\bullet \in \{ED,ER,EC\}$ as required and thus completes the proof.
\end{proof}

We remark that Proposition \ref{prop:network_bound_for_LS} does not hold for the two proximity measures based on vertex-deletions. In particular, there exist phylogenetic networks $N$ such that $d^{VD}(N) < d_{VD}(N)$. An example is the network $N_1$ depicted in Figure \ref{Fig_dvN_vs_dVLS}, where $d^{VD}(N_1)=6$, whereas $d_{VD}(N_1)=7$. Here, the network-based proximity measure has a smaller value than the $\LS$ graph-based one, because deleting the vertex incident with the edge leading to $x_5$ \enquote{breaks} the \enquote{interior} $K_4$. In $\LS(N_1)$ this vertex is not present anymore and thus cannot be deleted. In particular, in $\LS(N_1)$ one of the \enquote{outer} $K_4$'s has to be deleted completely in order to break the interior $K_4$. Based on this idea, it is in fact possible to construct non-edge-based phylogenetic networks, for which the difference between $d^{VD}(N)$ and $d_{VD}(N)$ is arbitrarily large. An example is depicted in Figure \ref{Fig_dvN_vs_dVLS_large}. Here, $d^{VD}(N)=6$ (again, deleting the vertex incident with the edge leading to leaf $x_5$ breaks the interior $K_4$), whereas $d_{VD}(N)=2m+8$ (because in order to break the interior $K_4$, one of the \enquote{arms} of $\LS(N)$ has to be deleted completely).
On the other hand, there exist phylogenetic networks $N$ such that $d^{VD}(N) > d_{VD}(N)$. For instance, consider the network $N_2$ depicted in Figure \ref{Fig_dvN_vs_dVLS}, where $d^{VD}(N_2)=8$, whereas $d_{VD}(N_2)=7$.

\begin{figure}[htbp]
    \centering
    \includegraphics[scale=0.3]{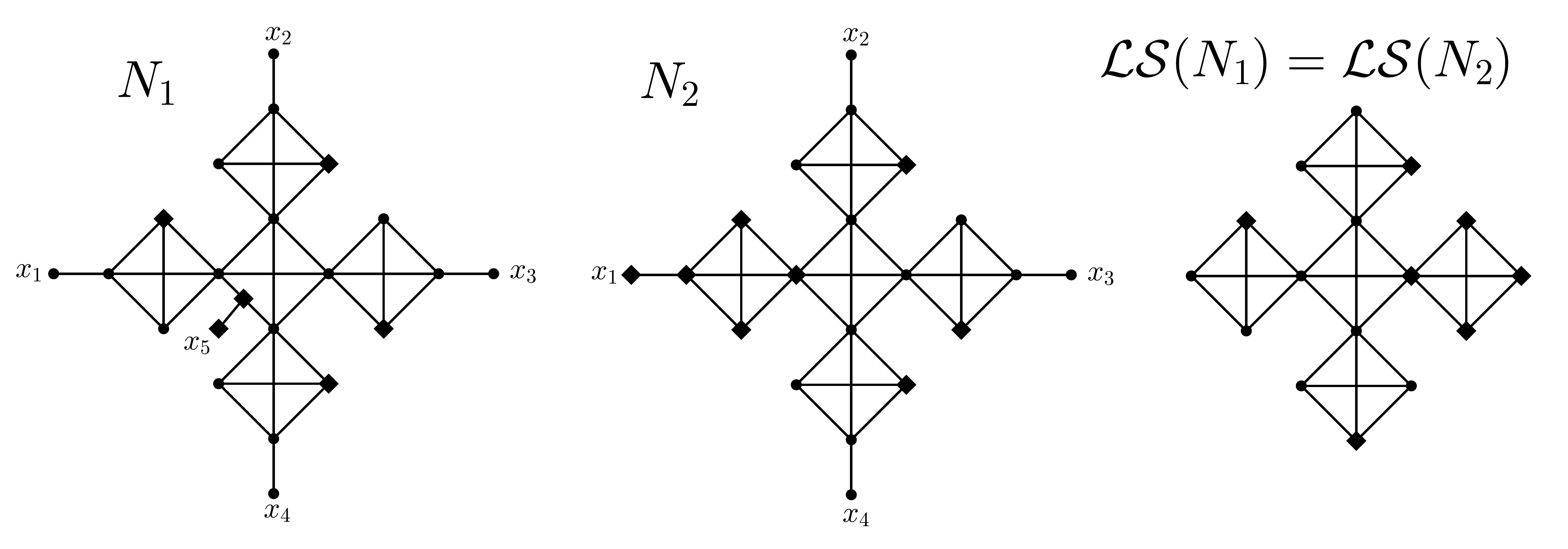}
    \caption{Phylogenetic network $N_1$ on $X=\{x_1, \ldots, x_5\}$ with $d^{VD}(N_1)=6$, whereas $d_{VD}(N_1)=7$, and phylogenetic network $N_2$ on $X=\{x_1, \ldots, x_4\}$ with $d^{VD}(N_2)=8$, whereas $d_{VD}(N_2)=7$. A possible choice of vertices to delete is given by the vertices depicted as diamonds in $N_1$, $N_2$, and $\LS(N_1)=\LS(N_2)$, respectively.}
    \label{Fig_dvN_vs_dVLS}
\end{figure}

\begin{figure}[htbp]
    \centering
    \includegraphics[scale=0.275]{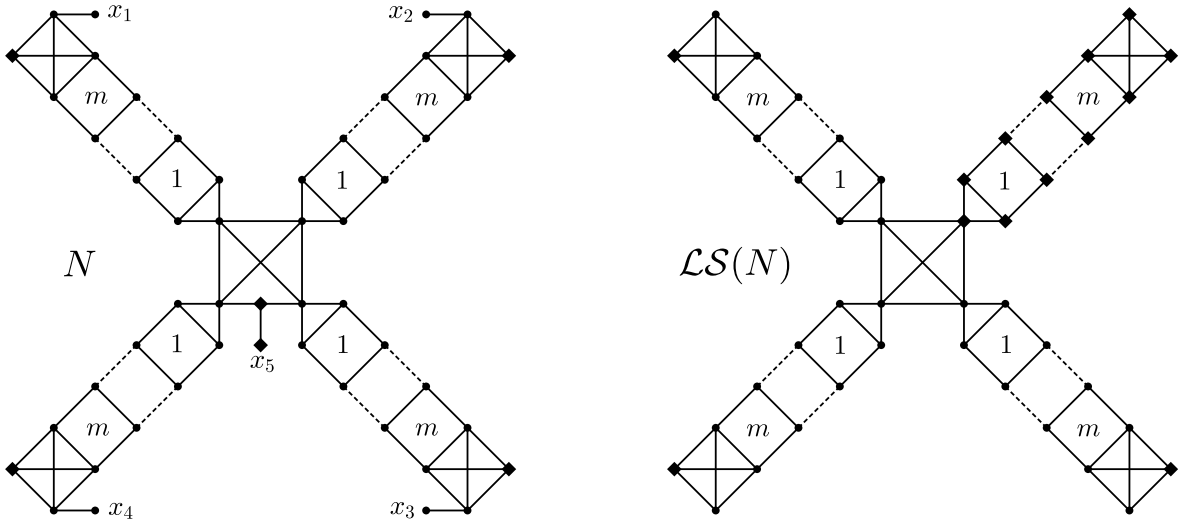}
    \caption{Phylogenetic network $N$ on $X=\{x_1, \ldots, x_5\}$ with $d^{VD}(N)=6$, whereas $d_{VD}(N)=2m+8$. Here, $m$ refers to the number of \enquote{squares} in each of the four \enquote{arms} of $N$, respectively $\LS(N)$, and a possible choice of vertices to delete is given by the vertices depicted as diamonds in $N$, respectively $\LS(N)$.}
    \label{Fig_dvN_vs_dVLS_large}
\end{figure}

Given the relatedness of $d_{\bullet}(N)$ and $d^\bullet(N)$ for $\bullet \in \{ED,ER,EC\}$ stated in Proposition \ref{prop:network_bound_for_LS}, it might seem redundant to consider both network-based and $\LS$ graph-based proximity measures. It turns out, however, that the network-based and $\LS$ graph-based proximity measures can induce different \enquote{rankings} of networks (where we rank networks in terms of their proximity to an edge-based graph). More precisely, for all types of proximity measures, there exist phylogenetic networks, $N_1$ and $N_2$ say, such that $d_\bullet(N_1) > d_\bullet(N_2)$ but $d^\bullet(N_1) < d^\bullet(N_2)$ (with $\bullet \in \{ED,ER,EC,VD\}$ fixed). Examples are given in Figures \ref{Fig_EdgeDeletion}, \ref{Fig_EdgeContraction}, and \ref{Fig_VertexDeletion}. This justifies considering proximity to edge-based graphs both on the level of the network as well as on the level of the $\LS$ graph.

\begin{figure}[htbp]
    \centering
    \includegraphics[scale=0.3]{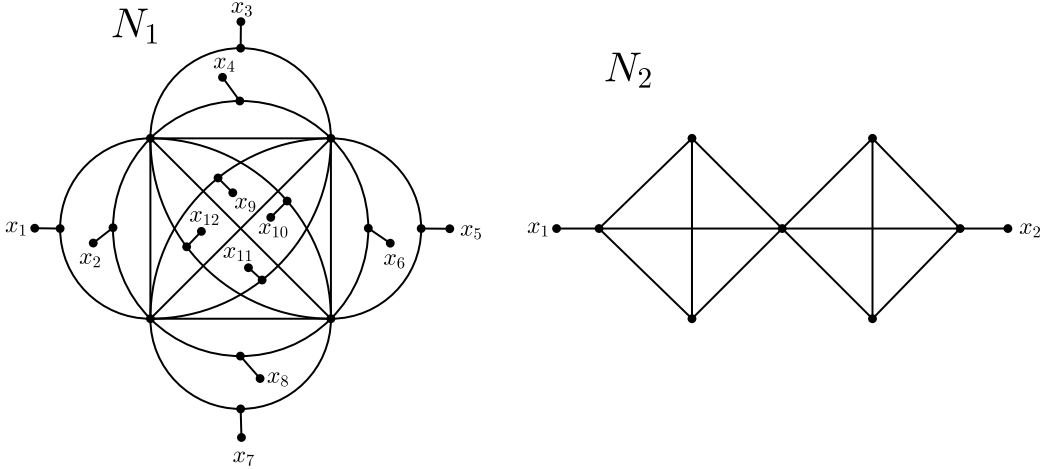}
    \caption{Phylogenetic networks $N_1$ and $N_2$. Note that $N_1$ can be constructed from a $K_4$ by copying each edge twice (to give three copies in total) and adding a leaf to each new copy. So in order to turn $N_1$ into a $K_4$-minor free graph, three edges need to be deleted; whereas because  $\LS(N_1)$ is isomorphic to $K_4$, only one edge needs to be deleted from $\LS(N_1)$. For $N_2$, there are basically two copies of $K_4$ that need to be broken. This leads to $d^{ED}(N_1)=d^{ER}(N_1)=3$, $d_{ED}(N_1)=d_{ER}(N_1)=1$, and $d^{ED}(N_2)=d^{ER}(N_2)=d_{ED}(N_2)=d_{ER}(N_2)=2$. In particular, $d^{ED}(N_1) > d^{ED}(N_2)$, whereas $d_{ED}(N_1) < d_{ED}(N_2)$ (and analogously $d^{ER}(N_1) > d^{ER}(N_2)$ and $d_{ER}(N_1) < d_{ER}(N_2)$).}
    \label{Fig_EdgeDeletion}
\end{figure}

\begin{figure}[htbp]
    \centering
    \includegraphics[scale=0.3]{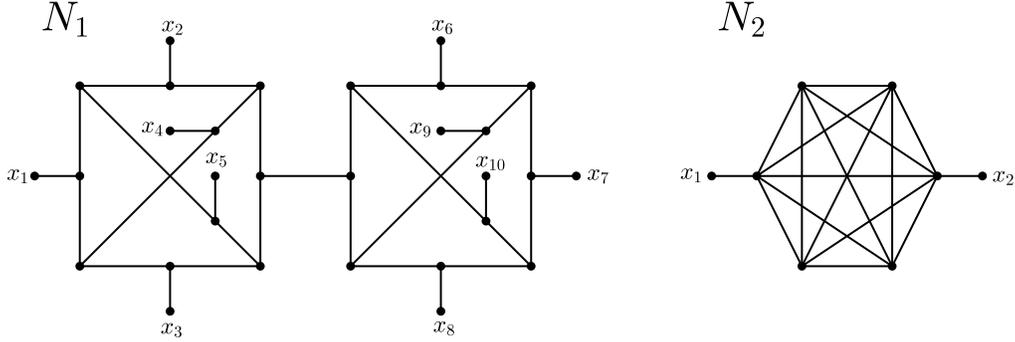}
    \caption{Phylogenetic networks $N_1$ and $N_2$.
    Note that $d^{EC}(N_1)=4$ because we have to \enquote{break} both $K_4$ minors in $N_1$ by contracting sufficiently many edges to turn each $K_4$ into a $K_3$. As each $K_4$ edge is subdivided, this is only possible by two contractions for each $K_4$ minor. Moreover, $d^{EC}(N_2)=3$, because the $K_6$ minor in the center can be reduced to $K_3$ by three contractions, which makes the graph edge-based; and fewer contractions are not sufficient to make the graph $K_4$-minor free. So, in summary, we get $d^{EC}(N_1)=4 > 3 = d^{EC}(N_2)$, whereas $d_{EC}(N_1)=2 < 3 = d_{EC}(N_2)$. The latter can again easily be seen as $\LS(N_2)$ is isomorphic to $K_6$, so we need to contract at least 3 edges to get to $K_3$ (cf. Corollary \ref{cor:LS_dZ_dV_equality}). $\LS(N_1)$, on the other hand, consists of two copies of $K_4$ connected by a cut edge, and it suffices to break each $K_4$ by contracting one edge each.}
    \label{Fig_EdgeContraction}
\end{figure}

\begin{figure}[htbp]
    \centering
    \includegraphics[scale=0.3]{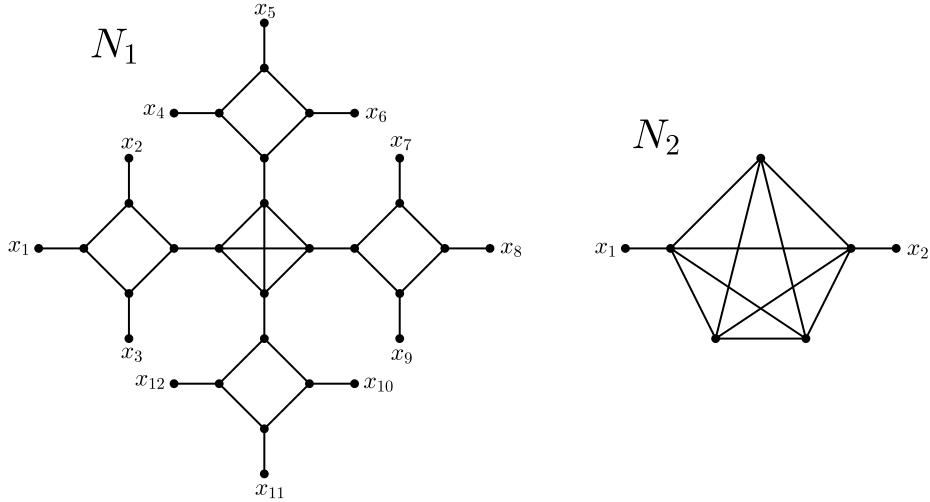}
    \caption{Phylogenetic networks $N_1$ and $N_2$. Here, in order to turn $N_1$ into a $K_4$-minor free graph, first one of the \enquote{arms} of the inner $K_4$ needs to be completely deleted (7 vertices), before a vertex from the $K_4$ itself can be deleted. However, $N_2$ clearly has $K_5$ as a minor, of which two vertices need to be deleted to make $N_2$ edge-based. This leads to $d^{VD}(N_1) = 8 > d^{VD}(N_2)=2$. On the other hand, $\LS(N_1)$ is isomorphic to $K_4$ and $\LS(N_2)$ is isomorphic to $K_5$, which shows that $d_{VD}(N_1) = 1 < 2 = d_{VD}(N_2)$.}
    \label{Fig_VertexDeletion}
\end{figure}

\subsection{Computational complexity of computing network-based and \texorpdfstring{$\LS$}{LS} graph-based proximity measures for edge-basedness}

In this section, we show that the calculation of all eight proximity measures introduced in the present manuscript is NP-hard.
We do so by using the results stated in Section  \ref{sec:computational_complexity} and by exploiting the fact that a connected graph (with $|V| \geq 2$) is edge-based if and only if it is $K_4$-minor free (cf. Proposition \ref{prop:K4-minorfree}).

\begin{Theo}
The calculation of $d^{VD}$ and $d_{VD}$ is NP-hard.
\end{Theo}

\begin{proof} 
We first show that the property $\pi \coloneqq$ \enquote{$K_4$-minor free} is hereditary on induced subgraphs, and non-trivial and interesting on connected graphs. If $G$ is $K_4$-minor free, then clearly any subgraph of $G$ is also $K_4$-minor free; in particular, $\pi$ is hereditary on induced subgraphs. Moreover, $\pi$ is non-trivial on connected graphs: for instance, every tree is connected and $K_4$-minor free, whereas all complete graphs with at least four vertices are not, i.e., there are connected graphs that are not $K_4$-minor free. Finally, $\pi$ is interesting on connected graphs since there are arbitrarily large connected graphs satisfying $\pi$ (e.g., arbitrarily large trees).

Now, by Theorem \ref{theo:vertex_deletion_NP}, this implies that given a connected simple graph $G$, it is an NP-hard problem to find a set of vertices of minimum cardinality whose deletion results in a connected and $K_4$-minor free subgraph of $G$. Importantly, by Proposition \ref{prop:K4-minorfree}, this implies that it is an NP-hard problem to find a set of vertices of minimum cardinality whose deletion results in an edge-based subgraph of $G$.

We next show that this particular connected vertex-deletion problem is also NP-hard for phylogenetic networks. In order to see this, take a connected simple graph $G$ and attach two additional leaves to each of its vertices, resulting in a phylogenetic network $N$. Now, if it was possible to efficiently find a set of vertices of minimum cardinality whose deletion results in an edge-based subgraph of $N$, the corresponding problem could also be solved efficiently for $G$; a contradiction. This is simply due to the fact that every triple consisting of an interior vertex of $N$, $u$ say, and its two attached leaves that needs to be deleted in $N$ to obtain an edge-based subgraph, corresponds to one vertex of $G$, namely $u$, that needs to be deleted to obtain an edge-based subgraph of $G$ (since the deletion of leaves in $N$ can only be necessary to keep the resulting graph connected, but not to destroy a subdivision of $K_4$ in $N$).

Finally, we show that the problem is also NP-hard when starting with an $\LS$ graph of some phylogenetic network. In order to see this, take a phylogenetic network $N$ and replace each of its vertices by a $K_4$ such that only one vertex of each new $K_4$ is incident to edges of $N$, i.e., we identify each vertex of $N$ with one vertex of a $K_4$ (cf. Figure \ref{Fig_GadgetVertices}). This leads to a graph $G$ that coincides with its own $\LS$ graph, as no leaves can be deleted and there are no degree-2 vertices or parallel edges. In fact, $G$ is an $\LS$ graph of some phylogenetic network, e.g., the one obtained from attaching a leaf to each of the newly added $K_4$'s. Now, if it was possible to efficiently find a set of vertices of minimum cardinality whose deletion results in an edge-based subgraph of $G$, the corresponding problem could also be solved efficiently for $N$. More precisely, each newly added $K_4$ in $G$ that needs to be deleted completely, corresponds to precisely one vertex of $N$ that needs to be deleted. If only one vertex in a newly added $K_4$ needs to be deleted in $G$, the corresponding vertex does not need to be deleted in $N$. Note that the key idea here is that in order to destroy a $K_4$, it is sufficient to delete one of its vertices. Thus, whenever a $K_4$ needs to be completely deleted, this indicates that this is necessary to keep the remaining graph connected.
Thus, there is a one-to-one correspondence between an optimal set of vertices to delete in $G$ and an optimal set of vertices to delete in $N$. If finding the former was easy, so would be the latter; a contradiction to the fact that the problem is NP-hard for phylogenetic networks.

Thus, both for phylogenetic networks and $\LS$ graphs finding a set of vertices of minimum cardinality whose deletion results in an edge-based graph is NP-hard. This implies that the calculation of $d^{VD}$ and $d_{VD}$ is NP-hard, too, which completes the proof.
\end{proof}

\begin{figure}[htbp]
    \centering
    \includegraphics[scale=0.35]{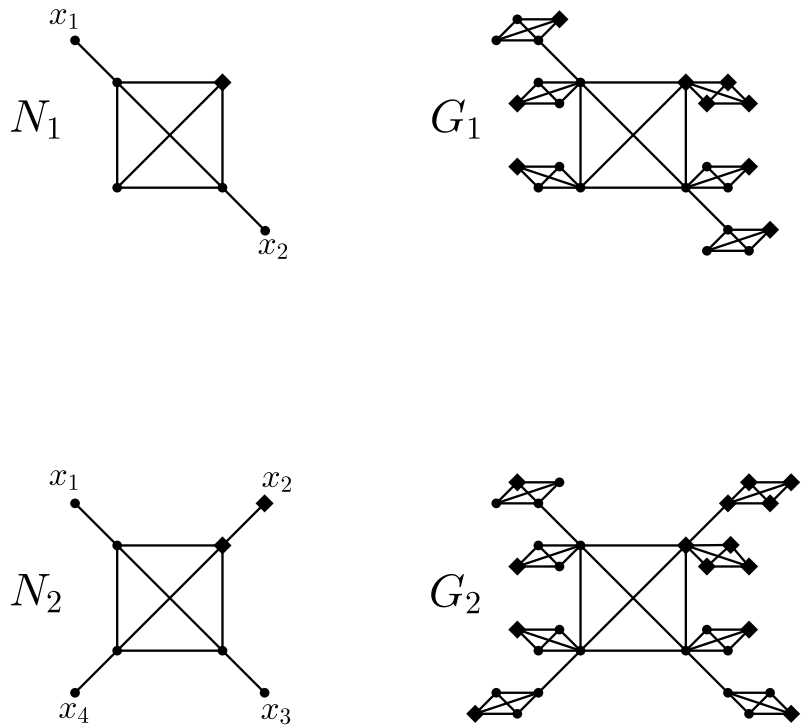}
    \caption{Two phylogenetic networks $N_1$ and $N_2$ and the graphs $G_1$, respectively $G_2$, obtained from them by identifying each vertex with a vertex of the complete graph $K_4$. In all cases, the vertices depicted as diamonds are part of a set of vertices of minimum cardinality whose deletion results in an edge-based graph.}
    \label{Fig_GadgetVertices}
\end{figure}

We now take a closer look at the other proximity measures.

\begin{Theo}
The calculation of $d^\bullet$ and $d_\bullet$ with $\bullet \in \{ED, ER, EC\}$ is NP-hard.
\end{Theo}
\begin{proof}
First recall that by Theorem \ref{N:E=U}, we have $d^{ED} = d^{ER}$, and by Corollary \ref{LS:E=U}, we have $d_{ED}=d_{ER}$. Therefore, it suffices to show the NP-hardness for $d^{EC}$ and $d_{EC}$ as well as for $d^{ED}$ and $d_{ED}$.

However, we begin by showing that both the edge-deletion and edge-contraction problems (see Section \ref{sec:computational_complexity}) for the graph property $\pi \coloneqq$ \enquote{$K_4$-minor free} are NP-hard.

For the edge-deletion problem, this follows directly from Theorem~\ref{Theo_Hardness_EdgeDeletion}, since in this case, the set of forbidden graphs consists precisely of $K_4$, i.e., $\mathbf{F}=\{K_4\}$, and $K_4$ is a simple and biconnected graph of minimum degree at least three.

For the edge-contraction problem, the statement follows from Theorem~\ref{Theo_Hardness_EdgeContraction} by noting that $\pi \coloneqq$ \enquote{$K_4$-minor free} satisfies conditions (C1)--(C4) therein. More precisely, $\pi$ is non-trivial on connected graphs as there are infinitely many connected graphs satisfying $\pi$ (e.g., trees of arbitrary size) and there are infinitely many connected graphs violating $\pi$ (e.g., the family of complete graphs $K_n$ with $n \geq 4$). Moreover, if a graph $G$ is $K_4$-minor free, then all contractions of $G$ are also $K_4$-minor free, and thus $\pi$ is hereditary on contractions. Furthermore, neither loops nor parallel edges influence whether a graph is $K_4$-minor free. In particular, a graph $G$ is $K_4$-minor free if and only if its simple graph is $K_4$-minor free, and thus $\pi$ is determined by the simple graph. Finally, clearly a graph $G$ is $K_4$-minor free if and only if its biconnected components are $K_4$-minor free (since $K_4$ is biconnected), and thus $\pi$ is determined by the biconnected components.

Thus, determining the minimum number of edge deletions or contractions required to turn an arbitrary graph into to a $K_4$-minor free graph must be NP-hard. Otherwise, the corresponding edge-deletion, respectively edge-contraction, problems would not be NP-hard, contradicting Theorem~\ref{Theo_Hardness_EdgeDeletion}, respectively Theorem~\ref{Theo_Hardness_EdgeContraction}.

We now first note that this also implies that calculating the distance to any $K_4$-minor free graph is NP-hard for \emph{simple} graphs. In order to see this, suppose that $G$ is a multigraph containing loops and/or parallel edges. We first note that loops do not influence the number of edge deletions or contractions required to turn a graph into a $K_4$-minor free graph (as deleting or contracting a loop can never help in destroying a subdivision of $K_4$). This implies that we can simply delete all loops of $G$ without changing the number of edge deletions or contractions required to turn it into a $K_4$-minor free graph. Thus, we may assume that $G$ is a loopless multigraph. Now, in case of edge contractions, it is clear that parallel edges do not influence the number of steps required to make $G$ $K_4$-minor free, either (since if we need to contract a parallel edge, $e=\{u,v\}$ say, all copies of $e$ are simultaneously contracted, and thus only one contraction is required). Thus, for each parallel edge of $G$, we can simply delete all copies but one, and obtain a simple graph $G'$ with the property that turning $G'$ into a $K_4$-minor free graph requires the same number of edge contractions as turning $G$ into a $K_4$-minor free graph does. Thus, if the edge-contraction problem could be solved efficiently for simple graphs, it could also be solved efficiently for multigraphs; a contradiction. Now, in case of edge deletions, we turn $G$ into a simple graph $G'$ by subdividing each copy of a parallel edge $e=\{u,v\}$ existing $k \geq 2$ times in $G$ with a degree-2 vertex $w_i$ for $i=1, \ldots, k$. However, this does not change the number of edge deletions required to reach a $K_4$-minor free graph: If some copy of $e=\{u,v\}$ needs to be deleted in $G$ to destroy a subdivision of $K_4$, it is sufficient to delete one of $e_1 = \{u,w_i\}$ or $e_2 = \{w_i,v\}$ in $G'$ for some $i$, and conversely, if for some $i$, one of $e_1 = \{u,w_i\}$ or $e_2 = \{w_i,v\}$ (where $w_i$ is a degree-2 vertex) needs to be deleted in $G'$ to destroy a subdivision of $K_4$, the corresponding copy of $e=\{u,v\}$ in $G$ needs to be deleted, too. Note that it cannot be the case that both $e_1$ and $e_2$ need to be deleted in $G'$ to obtain a $K_4$-minor free graph since $w_i$ is a degree-2 vertex, and thus after deleting one of $e_1$ and $e_2$, $w_i$ has degree-1 and its remaining incident edge is a cut edge and thus cannot be part of a subdivision of $K_4$. In particular, $G$ and $G'$ require the same number of edge deletions to turn them into $K_4$-minor free graphs. Thus, if the edge-deletion problem could be solved efficiently for simple graphs, it could also be solved efficiently for multigraphs; a contradiction.

Next we show that calculating the distance to any $K_4$-minor free graph is also NP-hard for connected graphs. If this was not true, we could efficiently solve the problem individually for each connected component and make each one of them $K_4$-minor free, which would give an efficient optimal solution for the general problem (since all graphs considered here are finite).

Taking the preceding two arguments together, we can additionally conclude that the problem is NP-hard for connected simple graphs.
If this was not the case, we could efficiently solve the problem for all simple graphs by solving it individually for each connected component and turning each of them into a $K_4$-minor free graph, yielding an optimal solution to the general problem; a contradiction to the fact that the problem is NP-hard for simple graphs.

We now show that determining the minimum number of edge deletions or contractions required to reach a $K_4$-minor free graph when starting with a phylogenetic network is also NP-hard. In order to see this, simply take a connected simple graph and attach an extra leaf to each of its vertices of degree larger than 1. This leads to a phylogenetic network; however, the newly added edges are all cut edges and thus do not have an impact on the number of edges that need to be deleted or contracted in order to reach a $K_4$-minor free graph. Therefore, if the problem could be solved efficiently for such networks, it could thus be solved efficiently for all connected simple graphs, which would be a contradiction.

Next, we show that determining the minimum number of edge deletions or contractions required to reach a $K_4$-minor free graph when starting with an $\LS$ graph of a phylogenetic network $N$ is also NP-hard. In order to see this, simply take a phylogenetic network and replace each of its leaves with a $K_4$ (cf. Figure \ref{Fig_GadgetEdges}). This leads to a graph $G$ that coincides with its own $\LS$ graph, as no leaves can be deleted and there are no degree-2 vertices or parallel edges. In fact, $G$ is an $\LS$ graph of some phylogenetic network, e.g., the one that we get by attaching a leaf to each of the newly added $K_4$'s (cf. Figure \ref{Fig_GadgetEdges}). Thus, if the minimum number of edge deletions or contractions in order to turn $G$ into a $K_4$-minor free graph could be efficiently calculated, we could also immediately calculate the number of such operations to turn the original network $N$ into a $K_4$-minor free graph. This is due to the fact that each newly added $K_4$ contributes precisely one required step (as these $K_4$'s each form a block, the number of edge deletions or edge contractions needed to turn $G$ into a $K_4$-minor free graph simply equals the number of added $K_4$'s, i.e., the number of leaves of $N$, plus the number of such operations needed to turn $N$ into a $K_4$-minor free graph). This would imply that the problem could be solved efficiently for all phylogenetic networks; a contradiction.

So calculating the minimum number of edge deletions or contractions to turn a graph (independent of whether it is simple or not), a connected (simple) graph, a phylogenetic network, or the $\LS$ graph of a phylogenetic network into a $K_4$-minor free graph is NP-hard. This is a major interim step, which we now use to show that the same is true for edge-basedness.

Note that in case we start with any connected graph, in order to reach a $K_4$-minor free graph, we never need to contract or delete a cut edge. This is due to the fact that a cut edge never belongs to any $K_4$-subdivision (as it does not even belong to a cycle), and thus its deletion or contraction will never be required to destroy any subdivision of $K_4$. Thus, in case we start with a connected graph, or, more specifically, with a phylogenetic network or its $\LS$ graph, the minimum number of steps to reach a $K_4$-minor free graph using edge deletions or contractions will always be achieved by going to a graph that is edge-based, i.e., a graph that is $K_4$-minor free \emph{and} connected. This shows that the calculation of  $d^{EC}$ and $d_{EC}$ as well as of $d^{ED}$ and $d_{ED}$ is indeed NP-hard, which completes the proof.
\end{proof}

\begin{figure}[htbp]
    \centering
    \includegraphics[scale=0.275]{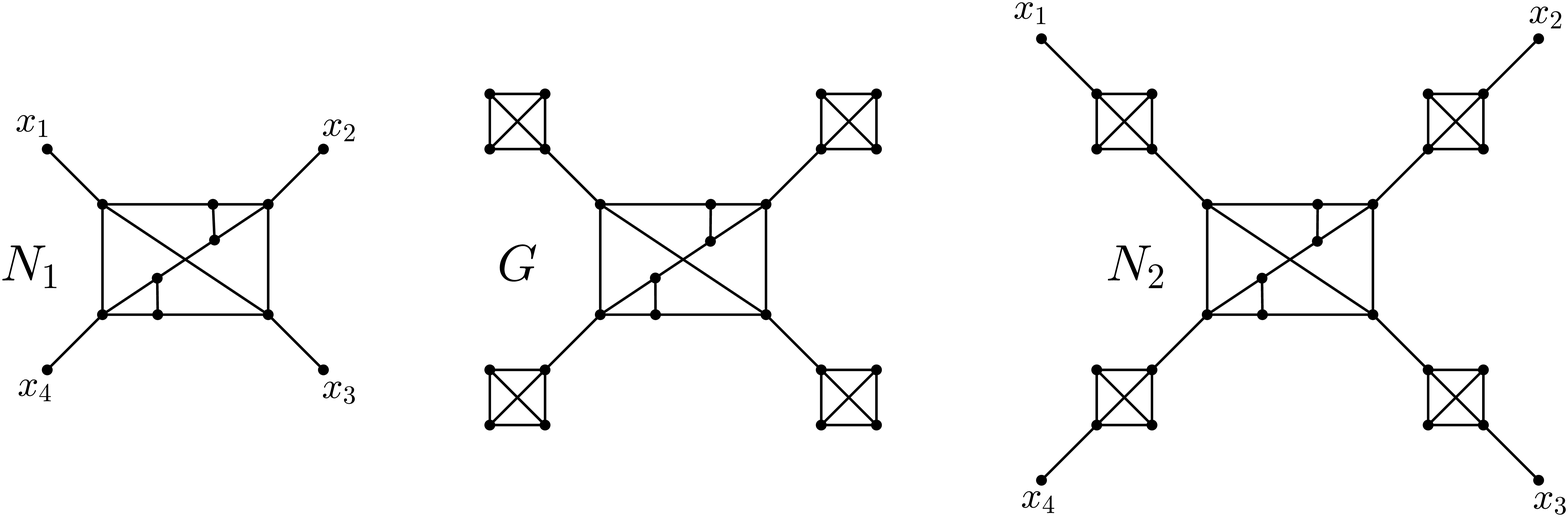}
    \caption{Graph $G$ obtained from the phylogenetic network $N_1$ by replacing each leaf of $N_1$ by the complete graph $K_4$. Note that $G$ is its own $\LS$ graph. Moreover, $G$ is the $\LS$ graph of the phylogenetic network $N_2$.}
    \label{Fig_GadgetEdges}
\end{figure}

\section{Discussion} 
In the present manuscript, we have introduced and analyzed two classes of proximity measures for edge-basedness of phylogenetic networks: The first one is based on the respective network itself, whereas the second one is based on its $\LS$ graph. Note that all of the measures we have introduced also work with general graphs and not only with phylogenetic networks, which might make them also interesting for graph theorists.

Furthermore, we have shown that some of these measures have substantially different properties, as they can lead to different rankings of networks -- i.e., they can differ in their decision concerning which one of two given networks is \enquote{closer} to being edge-based. Therefore, it is an interesting question for future research to determine which one of the introduced measures leads to biologically more plausible results. Moreover, it might be possible that another class of proximity measures operating on the blobs or blocks of $N$ or $\LS(N)$ instead of on the entire graph leads to better results; this will have to be investigated by future research.

As we have shown, the mere fact that deleting edges makes a network more likely to be edge-based, whereas adding edges makes it more likely to be tree-based, highlights that edge-basedness might be a concept that is biologically more relevant than tree-basedness. This is because networks with plenty of so-called reticulation events rarely occur in nature. Thus, edge-basedness and the distance of a network from it are very relevant for biological purposes.

We have also shown, however, that the concepts discussed here have an immense overlap with topics of classic graph theory. This makes edge-based networks also relevant for mathematicians. For instance, we have seen that while deciding whether a given network is edge-based is easy, it is generally NP-hard to determine the distance of a non-edge-based network to the nearest edge-based graph for all of the measures we have introduced. Thus, a very relevant problem for future research is to come up with good approximation algorithms for these measures. Last but not least, another intriguing mathematical challenge is finding good proximity measures to edge-basedness that can be calculated in polynomial time.

\section*{Acknowledgement} MF and TNH wish to thank the joint research project \textit{\textbf{DIG-IT!}} supported by the European Social Fund (ESF), reference: ESF/14-BM-A55-0017/19, and the Ministry of Education, Science and Culture of Mecklenburg-Vorpommerania, Germany. KW was supported by The Ohio State University’s President’s Postdoctoral Scholars Program.

\bibliographystyle{plainnat}
\bibliography{References}

\end{document}